\documentclass[a4paper,10pt,notitlepage]{article}

\usepackage[latin1]{inputenc}
\usepackage{amsfonts}
\usepackage{amsmath}
\usepackage{amssymb}
\usepackage{amsthm}
\usepackage{mathrsfs}
\usepackage{latexsym}
\usepackage[T1]{fontenc}
\usepackage[dvips]{graphicx}

\DeclareMathOperator{\supp}{supp}

\numberwithin{equation}{section}

\author{Sergio Albeverio\\ Institut f{\"u}r Angewandte Mathematik, Universit{\"a}t Bonn, D 53115, Bonn, Germany\footnote{SFB611; IZKS,Bonn; BiBoS,Bielefeld,Bonn; Acc. Arch. (US), Mandrisio}\\
Witold Karwowski\\Institute of Physics, Opole University, Opole, Poland}
\title{Jump processes on leaves of multibranching trees}
\date{\today}

\begin{document}

\newtheorem{defi}{Definition}[section]
\newtheorem{rem}[defi]{Remark}
\newtheorem*{rem*}{Remark}
\newtheorem{prop}[defi]{Proposition}
\newtheorem{cor}[defi]{Corollary}
\newtheorem{lem}[defi]{Lemma}
\newtheorem{thm}[defi]{Theorem}

\newcommand{\N}{\mathbb{N}}
\newcommand{\Z}{\mathbb{Z}}
\newcommand{\R}{\mathbb{R}}
\newcommand{\SB}{\mathbb{S}_B}
\newcommand{\E}{\mathcal{E}}
\newcommand{\K}{\mathcal{K}}
\newcommand{\Bmi}[2]{B_{\{\alpha^{#1}\}_{-(#2)}}}
\newcommand{\ami}[1]{{\{\alpha}\}_{-(#1)}}
\newcommand{\bmi}[1]{{\{\beta}\}_{-(#1)}}
\newcommand{\B}[1]{B_{\{\alpha\}_{#1}}}
\newcommand{\chimi}[1]{\chi_{\{\alpha\}_{-(#1)}}}
\newcommand{\chimibet}[1]{\chi_{\{\beta\}_{-(#1)}}}
\newcommand{\Bet}[1]{B_{\{\eta\}_{-(#1)}}}

\maketitle

\begin{abstract}
The p-adic numbers have found applications in a wide range of diverse fields of scientific research. In some applications the algebraic properties of p-adics enter as an indispensable ingredient of the theory. Another class of applications has to do with hierarchical tree like systems. In this context the applications are based on the well known correspondence between p-adics and the trees with $p$ branches emerging from every branching point. Then the algebraic structure does not enter and p-adics are used merely as a labeling system for the tree branches. We introduce a space of sequences denoted by ${\SB}$ suitable for labeling the trees with varying number of branches emerging from the branching points. We introduce a non Archimedean metric in ${\SB}$ and describe the basic topological properties of ${\SB}$. We also demonstrate that the known constructions of the stochastic processes on p-adics carry over to the stochastic processes on ${\SB}$ and hence on the corresponding trees.    
\end{abstract}

\section{Introduction}

Over the last twenty years a growing interest in p-adic numbers, see, e.g. \cite{Ko} and more generally in local fields can be observed also in connection with various applications beside those in number theory where they originated. In fact there is intensive research using p-adics in physics \cite{LiViM}, \cite{RTV}, \cite{VlaVoZ}, \cite{Kh1}, neural networks \cite{ATi}, cognitive systems \cite{Kh2}, \cite{Kh3} and in the theory of stochastic processes  \cite{E1}, \cite{E2}, \cite{AKa}, \cite{Koc}, \cite{DelTa}, \cite{KaViM}, \cite{AKaZ}, \cite{Kan}, \cite{Y} and other fields involving hierarchical systems \cite{ASTa}, \cite{Pa}. When analysing different lines of investigations one realizes that some of them like for instance applications in string theory \cite{VlaVoZ}, and in theory of p-adic distributions \cite{AKhSh1}, \cite{AKhSh2}, \cite{AKhSh3},\cite{Kh1}, \cite{Kh2}, \cite{AKhS4} rely on algebraic and topological properties of the local fields. On the other hand in such applications as the study of spin glasses \cite{RTV}, neural networks \cite{ATi}, and turbulence \cite{LiViM} the p-adic numbers appear as a well elaborated simple framework for describing a hierarchical structure. The algebraic properties of p-adics either do not enter into consideration at all, or play rather the role of a convenient but dispensable technical assumption.\\
There is a well known relation between the p-adics and a class of tree like graphs. Given a prime number $p$ the tree corresponding to the p-adic field $\mathbb{Q}_{p}$ has exactly $p$ branches emerging from every branching point. If a system under consideration displays a hierarchical structure, then it is natural to represent it by a tree, and then the trees related to p-adics appear as handy candidates. In such a situation p-adic numbers merely serve as labels for the tree branches. The advantage of using p-adics as the labels is that the distances between the ends of the branches are given by the p-adic metric.\\
In reality the structure of a hierarchical system may correspond to a tree which is more complicated than those labeled by p-adics. The number of branches emerging from branching points of the suitable tree may vary, and does not have to be a prime number. In order to carry over the investigations originally based on p-adic trees to the more complicated cases one needs an appropriate labeling system.\\
In this note we introduce such a system and argue that the results obtained for p-adics which do not rely on the algebraic structure of $\mathbb{Q}_{p}$ can indeed be extended to the more general class of trees. To support our claim we extend the results of \cite{AKa} to the case of general trees. It is remarkable that although the procedure is more complicated the main line of reasoning carry over without major changes, yielding essentially the same output of "rather explicit results" which was the main advantage of the approach presented in \cite{AKa} for the p-adics case. Let us remark that chaotic processes on trees have been studied by other means (e.g. as particular cases of stochastic processes on metric spaces \cite{A}, \cite{F}, \cite{FOT}, \cite{MaR}, \cite{Sturm}, and in connection with statistical mechanical problems like in percolation theory \cite{AlEv}). Our approach is more direct and permits e.g. to compute spectra of the generators.\\
The paper is organized as follows: In section 2. we introduce a class of sequence spaces denoted by ${\SB}$. The subscript B stands for a numerical function defined on sequences. A non Archimedean metric is defined on every ${\SB}$ and basic topological properties of ${\SB}$ are proved. We also establish a one to one correspondence between the spaces ${\SB}$ and the trees with a finite number of branches emerging from the branching points.\\
In section 3. we formulate and solve the Chapman-Kolmogorov equations and find the transition functions for a class of stochastic processes on ${\SB}$. In section 4. we consider the corresponding Markovian semigroup and provide a complete spectral description of their generators. We also give explicit formulae for the corresponding Dirichlet form.         

\section{The state space}

As usual we denote by $\Z$, $\N$ and $\N_0$ the set of integers, positive integers and
non-negative integers respectively. For any $k\in \Z$ let $S_k$ be the family of all  
sequences
$\{ \alpha_i\}_{i \leq k}$ such that $\alpha_i \in \N_0$ and $\alpha_i=0$ for all $i\leq N$ for some $N\leq k$, $N \in \Z$. Put
\[S=\bigcup_{k\in \Z} S_k.\]
Let $\alpha_{k+1}\in \N_0$. Then the product $\{\alpha_i\}_{i\leq k}\times\{\alpha_{k+1}\}$ can be identified with
$\{\alpha_i\}_{i\leq k+1}$,
\begin{equation}\label{2-1}
\{\alpha_i\}_{i\leq k}\times\{\alpha_{k+1}\}=\{\alpha_i\}_{i\leq k+1}.
\end{equation}
To simplify notations we put \[\{\alpha\}_k:=\{\alpha_i\}_{i\leq k}.\]
If $\alpha_i=0$ for all $i\leq k$ then we write $\{\alpha\}_k=\{0\}_k$.
\begin{defi}
Let $\B{k}$ be a function defined on $S$ with values in $\N\setminus \{1\}$.
\begin{enumerate}
\item We say that $\{\alpha\}_{k+1}$ is the $B$-product of $\{\alpha\}_k$ and $\{\alpha_{k+1}\}$ iff
\begin{equation}\label{2-2}
\{\alpha\}_{k+1}=\{\alpha\}_{k}\times \{\alpha_{k+1}\}
\end{equation}and
\begin{equation}\label{2-3}
0\leq \alpha_{k+1}\leq \B{k}-1.
\end{equation}
\item We say that $\{\alpha\}_{k+l},\ l\in \N$ is the $B$-product of $\{\alpha\}_k$ and the
ordered l-tuple $\{\alpha_{k+1},\dots ,\alpha_{k+l}\}$ 
\begin{displaymath}
\{\alpha\}_{k+l}=\{\alpha\}_k\times \{\alpha_{k+1},\dots ,\alpha_{k+l}\}
\end{displaymath}
iff
\begin{equation}\label{2-4}
\{\alpha \}_{k+l}=(\dots ((\{\alpha\}_k\times \{\alpha_{k+1}\})\times \{\alpha_{k+2}\})\times \dots \times \{\alpha_{k+l}\}),
\end{equation}
where all products are $B$-products in the sense of 1. We then write
\begin{equation}\label{2-5}
\{\alpha \}_{k+l}=\{\alpha\}_k\times \{\alpha_{k+1}\}\times \dots \times \{\alpha_{k+l}\}.
\end{equation}
\end{enumerate}
\qed\end{defi}
\begin{rem}
Whenever we write a formula like the right side of \eqref{2-5} we always mean the $B$-product.
\qed\end{rem}
\begin{defi}
Given a function $\B{k}$ as in Def. 2.1, we define $S_B \subset S$ by
\begin{enumerate}
\item $\{0\}_k \in S_B$ for all $k \in \Z$,
\item $\{\alpha \}_{k+1}\in S_B$ iff $\{\alpha\}_{k+1}=\{\alpha\}_k\times \{\alpha_{k+1}\},$ where $\{\alpha\}_k \in S_B$.
\end{enumerate}
\qed\end{defi}
It is easy to see that the following proposition holds:
\begin{prop}\label{p2-4}
$\{\alpha\}_k \in S_B$ iff there is $l\in \N$ such that
\begin{displaymath}
\{\alpha \}_k=\{0\}_{k-l}\times \{\alpha_{k-l+1}\}\times \dots \times \{\alpha_k\}.
\end{displaymath}
\qed\end{prop}
\begin{defi}
We say that a sequence $\{\alpha_i\}_{i\in \Z}$ belongs to the set $\SB$ iff $\{\alpha\}_k \in S_B$ for all $k\in \Z$.
\qed\end{defi}
To simplify notations we write
\begin{equation}\label{2-6}
\alpha :=\{\alpha_i\}_{i\in \Z}.
\end{equation}
Let $q$ be a real number, $q>1$. For any pair $\alpha ,\beta  \in \SB$ we define

\begin{equation}\label{2-7}\begin{split}
\rho_q(\alpha ,\alpha ) &= 0 \\
 \rho_q(\alpha ,\beta ) &= q^{-i_0},
\end{split}
\end{equation}

where $i_0$ is such that $\alpha_{i_0} \neq \beta_{i_0}$ and $\alpha_i = \beta_i$ if $i<i_0$. It is easy to see that the following proposition holds:
\begin{prop}\label{p2-6}
$\rho_q$ is a metric on $\SB$ satisfying the non-Archimedean triangle inequality
\begin{equation}\label{2-8}
\rho_q(\alpha ,\beta )\leq \max\{\rho_q(\alpha ,\gamma ),\rho_q(\gamma ,\beta )\}.
\end{equation}
\qed\end{prop}
It is clear that for any $q,q'>1$ the metrics $\rho_q$ and $\rho_{q'}$ are equivalent. Thus we fix a real number $q>1$ throughout the paper and drop the subscript $q.$ Set
\begin{equation}\label{2-9}
\mathbb{S}^k_B:=\{\alpha \in \SB;\alpha_i =0\ for\ i\geq k\},
\end{equation}
and
\begin{equation}\label{2-10}
\mathbb{S}_{B,0}:=\bigcup_{k\in \Z}\mathbb{S}^k_B.
\end{equation}
\begin{rem}\label{r2-7}
Since $(\SB,\rho )$ are going to constitute the basic state space of our stochastic processes it might be useful to have in mind an intuitive picture of it. $\SB$ is a space of sequences $\{\alpha_{k}\}$, which can be thought of as leaves at the ends of branches of a tree. The distance $\rho (\alpha,\beta)$ between two leaves $\alpha$ and $\beta$ can be thought of as the length of the branch from $\alpha$ to the nearest branching point connecting $\alpha$ and $\beta$ and then down to $\beta$ (where the tree is visualized vertically with its origin at the top). Precise relations between $\SB$`s and the trees will be provided later in this section. 
\end{rem} 
\begin{prop}\label{p2-8}
$\SB$ equipped with the metric $\rho$ is a complete metric space and $\mathbb{S}_{B,0}$ is a dense subset of it.
\end{prop}
\begin{proof} Let $\alpha^n \in \SB$, $n\in \N$ be a Cauchy sequence in the metric $\rho$. Then for any $k\in \N$ there is $N_k\in \N$ such that $\rho(\alpha^n,\alpha^m)<q^{-k}$ for all $n,m>N_k$. It follows that $\{\alpha^n\}_k=\{\alpha^m\}_k$. Without loss of generality we can assume $N_k$ to be increasing. For any $k\in \N$ choose $n_k>N_k$. Then
\begin{displaymath}
\{\alpha^{n_k}\}_k=\{\alpha^{n_{k+1}}\}_k
\end{displaymath}
and
\begin{equation}\label{2-11}
\{\alpha^{n_{k+1}}\}_{k+1}=\{\alpha^{n_k}\}_k\times \{\alpha^{n_{k+1}}_{k+1}\}.
\end{equation}
Define $\alpha=\{\alpha_i\}_{i\in \Z}$ by
\begin{equation}\label{2-12}
\{\alpha \}_l=\{\alpha^{n_1}\}_l\ \textrm{for}\ l\leq 1
\end{equation}
and
\begin{equation}\label{2-13}
\{\alpha\}_l=\{\alpha^{n_1}\}_1\times \{\alpha^{n_2}\}_2\times \dots \times \{\alpha^{n_l}_l\}=\{\alpha^{n_l}\}_l\ \textrm{for}\ l>1.
\end{equation}
Since $\alpha^n \in \SB$ it follows from \eqref{2-11}, \eqref{2-12}, \eqref{2-13} that $\alpha \in \SB$. Let $n>N_k$. Then $\rho (\alpha^{n_k},\alpha^n)<q^{-k}$ and by \eqref{2-13} $\rho(\alpha ,\alpha^{n_k})\leq q^{-(k+1)}$. Consequently $\rho(\alpha, \alpha^n) \leq \max \{\rho(\alpha, \alpha^{n_k}),\rho(\alpha^{n_k},\alpha^n)\}\leq q^{-k}$, which shows that $\alpha^n$ converges to $\alpha$ in the metric $\rho$. Thus $\SB$ is complete.
Let now $\alpha \in \SB$. Define $\alpha^n=\{\alpha\}_n\times \{0\}\times \dots$.  Clearly $\alpha^n \in \mathbb{S}_{B,0}$ and $\alpha^n\to \alpha$ in the metric $\rho$ as $n\to \infty$, which proves that $\mathbb{S}_{B,0}$ is dense in $\SB$.
\end{proof} 

Given $\alpha \in \SB$ and $N \in \Z$ the set
\begin{equation}\label{2-14}
K(\alpha ,q^N)=\{\beta \in \SB;\rho(\alpha, \beta)\leq q^N\}
\end{equation}
will be called a ball of radius $q^N$ centered at $\alpha$. As consequences of \eqref{2-8} we have
\begin{enumerate}
\item If $\beta \in K(\alpha, q^N)$ then $K(\beta ,q^N)=K(\alpha ,q^N)$.
\item $K(\alpha , q^N),K(\beta ,q^N)$ are either disjoint or identical, for any $\alpha,\beta \in \SB$.
\item If $\alpha \in \SB$ and $\alpha_i=0$ for $i<-N$, then $K(\alpha ,q^N)=K(0,q^N)$.
\end{enumerate}
It follows from \eqref{2-14} that $K(\alpha ,q^N)$ is uniquely defined by $\{\alpha\}_{-(N+1)}$ and thus we can identify
\begin{equation}
\label{2-15}
\{\alpha\}_{-(N+1)}=K(\alpha,q^N).
\end{equation}
With this notation one easily sees that
\begin{enumerate}
\addtocounter{enumi}{3}
\item 
\begin{equation}\label{2-16}
K(\alpha ,q^{N+1})=\{\alpha\}_{-(N+2)}=\bigcup_\gamma\{\alpha\}_{-(N+2)}\times \{\gamma\}.
\end{equation}
\end{enumerate}
By construction the union is taken over the values of $\gamma$ satisfying: $0\leq \gamma \leq \Bmi{}{N+2}-1$. Thus the ball $K(\alpha,q^{(N+1)})$ is the union of $\Bmi{}{N+2}$ disjoint balls of radius $q^N$. Take $N,M \in \Z$, $N>M$. Iterating formula \eqref{2-16} we find a family of disjoint balls of radius $q^M$ such that $K(\alpha ,q^N)$ can be expressed as their union. When the function $\B{k}$ is defined, this family depends on $\alpha \in \SB$ and the numbers $N,M$. We denote this family by $\K(\alpha ,N,M)$ and denote by $n(\alpha ,N,M)$ the number of balls in it. Note that $\K(\alpha ,N,M)\subsetneq \K(\alpha ,N+1,M)$. Consequently $n(\alpha ,N,M)$ increases to infinity as $N$ varies from $M+1$ to $+\infty$.\\
Let $M\in \Z$ be given. Then according to 3) for any $\beta \in \SB$ there is $N>M$ such that $\beta \in K(0,q^N)$. Thus
\begin{equation}\label{2-17}
\SB=\bigcup_{N>M}K(0,q^N)
\end{equation}
On the other hand $\beta \in K(0,q^N)$ implies that $\beta$ belongs to one of the balls in the family $\K(0,N,M)$. Set
\begin{equation}\label{2-18}
\K(M):=\bigcup_{N>M}\K(0,N,M).
\end{equation}
Then $\K(M)$ is a countable family of disjoint balls of radius $q^M$
\begin{equation}\label{2-19}
\K(M)=\{K^M_i\}_{i\in \N},
\end{equation}
where $K^M_i$ is a ball of radius $q^M$ and $K^M_i\cap K^M_j=\emptyset$ iff $i\neq j$. As a consequence of \eqref{2-16}, \eqref{2-17} we have
\begin{equation}\label{2-20}
\SB=\bigcup_{i\in \N}K^M_i.
\end{equation}
We conclude this section with following observations. Let $p$ be a prime number and define
$\B{k}=p$ for all $\{\alpha\}_k \in S$. Then $\SB$ is identical with the set of $p$-adic numbers $\mathbb{Q}_p$.
If we put $q=p$ then the $\SB$ metric coincides with the $\mathbb{Q}_p$ metric. It is well known that any $p$-adic ball is both open and compact. The same is true for the balls in $\SB$ in general.
\begin{prop}\label{p2-9}
A ball in $\SB$ is both open and compact.
\end{prop}
\begin{proof}
Let $\alpha \in \SB$ and $k\in\Z$. Then
$$
\ami{k+1} = \left\lbrace \beta \in \SB; \rho (\alpha ,\beta)\leq q^k\right\rbrace=\left\lbrace \beta \in \SB; \rho (\alpha ,\beta) < q^{k+1}\right\rbrace.
$$
According to the general definition of topology in metric spaces the right hand side denotes an open ball in $\SB$. Thus $\ami{k+1}$ is open.\\
 Consider a sequence $\beta^n,n\in\N$ of elements belonging to the ball $\ami{k+1}$.\\
Since $\ami{k+1}=\cup_{\gamma_1} \ami{k+1} \times \{\gamma_1\}$, where $0\leq\gamma_1 \leq \Bmi{}{k+1}-1$, there is a value of $\gamma_1$ such that infinitely many elements of the sequence $\beta^n$ belong to the ball $\ami{k+1}\times \{\gamma_1\}$. We choose one of these elements and denote it by $\beta^{n_1}$. Iterating this procedure we obtain a descending sequence of balls $\ami{k+1}\times \{\gamma_1\}\times\{\gamma_2\}\times \dotsm \times \{\gamma_i\}$ and a subsequence of $\beta^n:$\\
$$
\beta^{n_i} \in \ami{k+1} \times \{\gamma_1\}\times \dotsm \{\gamma_i\}.
$$
The intersection $\cap_{i\in\N} \ami{k+1}\times \{\gamma_1\}\times \dotsm \times \{\gamma_i\}$ contains exactly one element
$
\beta = \left\lbrace \dotsc ,\alpha_{-(k+2)},\alpha_{-(k+1)},\gamma_1,\gamma_2,\dotsc \right\rbrace
$
and clearly $\beta^{n_i} \to \beta$ as $i \to \infty$. Hence by the Bolzano-Weierstrass theorem the ball $\ami{k+1}$ is compact.
\end{proof}
Put $\{\mathbf{x}^{n}\}_{n\in{\Z}}$ for a sequence of points in $\R^2$, then we have $(\mathbf{x}^{n}\in{\R^{2}})$.
Define a family $\mathcal{X}$ of sequences by
\begin{equation}\label{2-21}
 \mathcal{X} = \{X\subset{\R}^{2}; X=\{\mathbf{x}^{n}\}_{n\in{\Z}},\left|\mathbf{x}^{n+1}-\mathbf{x}^{n}\right|=2^{-(n+3)}\}.
 \end{equation}
If $X\in{\mathcal{X}}$ then the set
\begin{equation}\label{2-22}
l_{X}=\bigcup_{n\in{\Z}}\{s\mathbf{x}^{n+1}+(1-s)\mathbf{x}^{n}, 0<s\leq 1\}, X=\{\mathbf{x}^{n}\}_{n\in{\Z}}
\end{equation}
is a continuous planar line composed of segments with the ends $\mathbf{x}^{n}$, $\mathbf{x}^{n+1}.$
We shall define a tree as the union of the lines $l_{X}$
\begin{equation}\label{2-23}
\bigcup_{X\in{\mathcal{T}}}l_{X}.
\end{equation}
The set $\mathcal{T}\subset{\mathcal{X}}$ is defined as follows.
Let $X_{0}=\{{\mathbf{x}}_{0}^{n}\}_{n\in{\Z}}\in{\mathcal{X}}$ be such that $l_{X_{0}}$ is a vertical semi axis directed upright. Then
\begin{enumerate}
\item
$X_{0}\in{\mathcal{T}}$,
\item 
If $X\in{\mathcal{T}}$ then there is $N\in{\Z}$ such that $\mathbf{x}^{n}=\mathbf{x_{0}}^{n}$ for all $n<N$. Moreover
$\mathbf{x}^{n}$, $n\in{\Z}$ is located either on the line $l_{X_{0}}$ or to the right from $l_{X_{0}}.$
\end{enumerate}
For any $X\in{\mathcal{T}}$ and $M\in{\Z}$ define
\begin{equation}\label{2-24}
\mathcal{F}_{XM}=\{Y\in{\mathcal{T}};\mathbf{y}^{n}=\mathbf{x}^{n}, n<M\}.
\end{equation}
\begin{enumerate}
\addtocounter{enumi}{2}
\item
Given $X\in{\mathcal{T}}$ and $M\in{\Z}$ there are $k\in{\N}$ and pairwise different points $\mathbf{y}_{0}=\mathbf{x}^{M},\mathbf{y}_{1}$,...,$\mathbf{y}_{k}$ such that for any $i\in{\{0,1,...,k\}}$
there is $Y\in{\mathcal{F}_{XM}}$ such that $\mathbf{y}^{M}=\mathbf{y}_{i}.$ Conversely if $Y\in{\mathcal{F}_{XM}}$  then $\mathbf{y}^{M}=\mathbf{y}_{i}$ for some $i\in{\{0,1,...,k\}.}$
\end{enumerate}
Under conditions 1-3 the identification of a tree with \eqref{2-23} is consistent with the intuitive picture of a tree. For any $X\in{\mathcal{T}}$ and $n\in{\Z}$ there is a branching at $\mathbf{x}^{n}.$ The number of branches emerging from $\mathbf{x}^n$ is finite and not less than 2. Note that the conditions 1), 2), 3) determine the tree structure rather than the exact locations of its branches.\\
Let a space $\SB$ be given as above. We choose the metric $\rho_{q}$ with $q=2$. We shall construct a tree by defining an injection $L:\SB\longrightarrow\mathcal{X}.$ It will be required that if $L\alpha = X,\alpha \in{\SB}$ then  
$\mathbf{x}^{n+1}(\alpha)=\mathbf{x}^{n+1}(\{\alpha \}_{n+1}).$\\
Then the map $L:\SB\longrightarrow\mathcal{X}$ is defined recursively by
\begin{itemize}
\item
$$L\{0\}=X_{0}$$
\item
If $X=L\alpha$, $\alpha\in{\SB}$ then $\mathbf{x}^{n}$, $n\in{\Z}$ is a branching point with $\Bmi{}{n}$ branches emerging from it. 
The points 
\begin{equation}\label{2-25}
\mathbf{x}^{n+1}(\{\alpha\}_{n},0),\mathbf{x}^{n+1}(\{\alpha\}_{n},1),
...,\mathbf{x}^{n+1}(\{\alpha\}_{n},\Bmi{}{n}-1)
\end{equation}
 are located in the order \eqref{2-25} from left to right and satisfy
\begin{equation}\label{2-26}
\left|\mathbf{x}^{n}(\{\alpha\}_{n-1},\alpha_{n})-\mathbf{x}^{n+1}(\{\alpha\}_{n},i)\right|=2^{-(n+3)}, 
 i=0,1,...,\Bmi{}{n}.
\end{equation}
\end{itemize}
Thus every $\alpha\in{\SB}$ defines a sequence $X=\{\mathbf{x}^{n}\}_{n\in{\Z}}.$ 
Due to \eqref{2-26} the sequence $\mathbf{x}^{n}(\{\alpha\}_{n},\alpha_{n+1})$ converges in the ${\R}^{2}$ 
norm as $n\rightarrow\infty$ to a point $\mathbf{x}.$ We shall say that $\mathbf{x}$ is a leave at the end of the
branch $l_{X}.$ The distance from $\mathbf{x}^{n}$ to $\mathbf{x}$ along $l_{X}$ equals $2^{-(n+2)}.$ Let $X,Y\in\mathcal{T}$ and $\mathbf{x},\mathbf{y}\in{\R^{2}}$ be the leaves of $l_{X}$ and $l_{Y}$ respectively. If $n_{0}$ is the maximal value of $n\in{\Z}$ such that $\mathbf{x}^{n_{0}}\in{l_{X}\cap l_{Y}}$ then the distance $\rho_{tree}(\mathbf{x},\mathbf{y})$ between $\mathbf{x}$ and $\mathbf{y}$ is defined as the distance from $\mathbf{x}$ to $\mathbf{x}^{n_{0}}$ along $l_{X}$ plus the distance from $\mathbf{x}^{n_{0}}$ to $\mathbf{y}$ along $l_{Y}$ hence it equals $2^{-(n_{0}+1)}$:
\begin{equation}\label{2-27}
\rho_{tree}(\mathbf{x},\mathbf{y})=2^{-(n_{0}+1)}.  
\end{equation}
If $X=L\alpha$ and $Y=L\beta$ then $\{\alpha\}_{n_{0}}=\{\beta\}_{n_{0}}$ 
and $\alpha_{n_{0}+1}\neq\beta_{n_{0}+1}.$ Thus $\rho (\alpha,\beta)=2^{-(n_{0}+1)}.$ To summarize
if $X=L\alpha$, $Y=L\beta$ and $\mathbf{x}$,$\mathbf{y}$ are the leaves of $l_{X}$,$l_{Y}$ respectively then
\begin{equation}\label{2-28}
\rho_{tree}(\mathbf{x},\mathbf{y})=\rho (\alpha,\beta).
\end{equation}
In this note we discuss random processes with ${\SB}$ as the state space. By the above discussion it is equivalent 
to say that the state space is the set of branches $l_{X}$ or the set of leaves 
$\mathbf{x}=\lim_{n\rightarrow \infty} \mathbf{x}^{n}$. This observation explains why in the title of the paper we used the concept of "leaves on multibranching trees". 

\section{Stochastic processes on $\SB$}
In this section we shall construct a class of stochastic processes on $\SB$. The main step in this direction will be a construction of the processes on $\K(M)$. Let $\alpha^i \in K^M_i$, $i\in \N$. Then according to \eqref{2-15} 
 
\begin{equation}\label{3-1}
K^M_i=K(\alpha^i,q^M)=\{\alpha^i\}_{-(M+1)}.
\end{equation}
Put $P_{\{\alpha^i\}_{-(M+1)}\{\alpha^j\}_{-(M+1)}}(t),t\in\R_+$ for the transition probability from $K^M_i$ to $K^M_j$ in time t. Whenever possible we shall use the simplified notation
\begin{equation}\label{3-2}
P^{M,M}_{ij}:=P_{\{\alpha^i\}_{-(M+1)}\{\alpha^j\}_{-(M+1)}}(t).
\end{equation}
Thus the forward and backward Chapman-Kolmogorov equations read:
\begin{subequations}\label{KC}
\begin{equation}
\label{KCa}
\dot{P}^{M,M}_{ij}(t)=-\tilde{a}_jP^{M,M}_{ij}(t)+\sum^{\infty}_{\substack{l=1\\ l\neq j}}{\tilde{u}_{lj}P^{M,M}_{il}},
\end{equation}\begin{equation}
\label{KCb}
\dot{P}^{M,M}_{ij}(t)=-\tilde{a}_jP^{M,M}_{ij}(t)+\sum^{\infty}_{\substack{l=1\\ l\neq j}}{\tilde{u}_{il}P^{M,M}_{lj}},
\end{equation}
\end{subequations}
$i,j\in\N$. We impose the initial condition
\begin{equation}\label{3-4}
P^{M,M}_{ij}(0)=\delta_{ij}.
\end{equation}

 The coefficients $\tilde{a}_j$ and $\tilde{u}_{lj}$ will be defined according to the following intuitive requirements.
If the process is in the ball $K(\alpha,q^N)$ at time $t$ then, for small real positive $\Delta t$, the probability that at time $t+\Delta t>t$ the process is outside of $K(\alpha,q^N)$ is set equal to $a(\alpha,N)\Delta t$. We call $a(\alpha,N)$ the intensity of the state $K(\alpha,q^N)$ and assume
\begin{enumerate}
 \renewcommand{\theenumi}{(\roman{enumi})}
\item\label{3i}The intensity of the state $K(\alpha ,q^N)$ depends only on the radius of the ball, i.e. on N, and is independent of $\alpha$. Hence
$$P \left( X_{t+\Delta t}\in (\SB - K(\alpha, q^N))\vert X_t\in K(\alpha,q^N)\right)=a(N)\Delta t.$$

\item\label{3ii} The probability that during the short time $\Delta t$ the process jumps over a distance $q^N$ and reaches a state in $K(\alpha, q^{N'})$ is the same as the probability to reach a state in $K(\beta, q^{N'})$ where $M\leq N'<N$ and
 \mbox{$\rho(\alpha ,\beta )=q^{N'+1}.$}
\item\label{3iii} The coefficients $\tilde{a}_{j}$ satisfy the following relation: 
\begin{equation}\label{3-5}
\tilde{a}_j=\sum^{\infty}_{\substack{l=1\\l\neq j}}\tilde{u}_{jl}.
\end{equation}
\end{enumerate}
To meet requirement \ref{3i} we proceed as follows. We define a sequence $a(N)$, $N \in \Z$ such that
\begin{subequations}
\begin{equation}
a(N)\geq a(N+1)
\end{equation}
and
\begin{equation}
\lim_{N\to \infty}a(N)=0,\ \lim_{N\to -\infty}a(N)=W,
\end{equation}
\end{subequations}
where $W$ is either a positive real number or $+\infty$.
Put
\begin{equation}\label{3-7}
U(N+1)=a(N)-a(N+1).
\end{equation}
Then $U(N+1)\Delta t$ is the probability that the process leaves a ball $K(\alpha , q^N)$ but stays in the ball $K(\alpha ,q^{N+1})$ i.e. it jumps to one of the balls
\begin{equation}\label{3-8}
\{\alpha\}_{-(N+2)}\times \{\gamma \},\ \gamma=0,\dotsc ,\Bmi{}{N+2}-1,\ \gamma \neq \alpha_{-(N+1)}.
\end{equation}
Let $\rho(\alpha^l,\alpha^j)=q^{N+m}$, $m\in \N$. Define
\begin{equation}\label{3-9}
B(\alpha^j,m,M)=(\Bmi{j}{M+m+1}-1)\Bmi{j}{M+m}\dotsm \Bmi{j}{M+2}
\end{equation} 
Set
\begin{equation}\label{3-10}
u(\alpha^j,m,M):=B^{-1}(\alpha^j,m,M)U(M+m).
\end{equation}
It follows from \ref{3i}, \ref{3ii} that $u(\alpha^{j},m,M)\Delta t$ is the probability that the process
jumps from $\{\alpha^l\}_{-(M+1)}$ to $\{\alpha^j\}_{-(M+1)}$ during the time $\Delta t$. Thus we define
\begin{equation}\label{3-11}
\tilde{u}_{lj}=u(\alpha^j,m,M).
\end{equation}
To underline the fact that the elementary balls have radius $q^M$ we write
\begin{equation}\label{3-12}
\tilde{a}_j=\tilde{a}_j(M).
\end{equation}

\begin{lem}\begin{equation}\label{3-13}
\tilde{a}_j(M)=a(M).
\end{equation}
\end{lem}
\begin{proof} According to \eqref{3-5} we have $\tilde{a}_j(M)=\sum_{l=1}^{\infty}\tilde{u}_{jl}$. We shall compute the right hand side with $\tilde{u}_{lj}=u(\alpha^j,m,M)$ as defined by \eqref{3-10}, \eqref{3-11}. We have
\begin{equation}\label{3-14}
\sum_{\substack{l=1\\l\neq j}}^{\infty}\tilde{u}_{lj}=\sum_{m=1}^{\infty}\left( \sum_lu(\alpha^l,m,M)\right) ,
\end{equation}
where the summation in the bracket runs over $l\in \mathbb{N}$ such that $\rho(\alpha^j,\alpha^l)=q^{M+m}$. Then
\begin{equation}\label{3-15}
\{\alpha^l\}_{-(M+1)}=\{\alpha^j\}_{-(M+m+1)}\times \{\gamma_m\}\times \dotsm \times \{\gamma_1\}
\end{equation}
where the admissible m-tuples $\{\gamma_m,\dotsc ,\gamma_1\}$ are such that \eqref{3-15} is a B-product. Moreover $\gamma_m\neq\alpha^j_{-(M+m)}$. Hence by \eqref{3-9} and \eqref{3-10}
\begin{displaymath}
\sum_lu(\alpha^l,m,M)=\sum_{\gamma_m\neq \alpha^j_{-(M+m)}}\sum_{\gamma_{m-1}}\dotsb \sum_{\gamma_1}u(\alpha^l,m,M)=U(M+m)
\end{displaymath}
and
\begin{displaymath}
\tilde{a}_j(M)=\sum_{m=1}^{\infty}\left(\sum_lu(\alpha^l,m,M)\right)=\sum_{m=1}^{\infty}U(M+m)=a(M).
\end{displaymath}\end{proof} 
Our next observation follows directly from \eqref{3-9} and \eqref{3-10}. Namely if $m\geq 2$ then
\begin{equation}\label{3-16}
\Bmi{j}{M+2}u(\alpha^j,m,M)=u(\alpha^j,m-1,M+1).
\end{equation}
Let us turn to the problem of solving equations \eqref{KC}. For this we need some preparations. According to \eqref{3-9}-\eqref{3-11} $\tilde{u}_{lj}$ depend explicitly on $\{\alpha^j\}_{-(M+2)}$ and thus $\tilde{u}_{lj}$ is the same for all target states $\{\alpha^j\}_{-(M+2)}\times \{\gamma\}.$ The dependence on the initial state enters only via the distance $\rho(\alpha^l,\alpha^j)$. Thus $\tilde{u}_{lj}=u(\alpha^j,m,M)$ is the same for all $l$ such that
\begin{equation}\label{3-17}
\rho(\alpha^l,\alpha^j)=q^{M+m},
\end{equation}
and we can write \eqref{KCa} in the form
\begin{equation}\label{3-18}
\dot{P}^{M,M}_{ij}(t)=-a(M)P^{M,M}_{ij}+\sum_{m=1}^{\infty}u(\alpha^j,m,M)\sum_lP^{M,M}_{il},
\end{equation}
where the index $l$ in $\sum_lP^{M,M}_{il}$ satisfies \eqref{3-17} for the corresponding value of $m$. In view of \eqref{2-16} and \eqref{3-15} we have
\begin{equation}\label{3-19}
\bigcup_{l,0\leq \rho(\alpha^l,\alpha^j)\leq q^{M+m}}\left\lbrace \alpha^l\right\rbrace =\left\lbrace \alpha^j\right\rbrace{}_{-(M+m+1)},
\end{equation}
which suggests the notation
 \begin{multline}\label{3-20}
P^{M,M+m}_{ij}:=P_{\{\alpha^i\}_{-(M+1)}\{\alpha^j\}_{-(M+m+1)}}\\
:=\sum_{l,0\leq \rho(\alpha^l,\alpha^j)\leq q^{M+m}}P_{\{\alpha^i\}_{-(M+1)}\{\alpha^l\}_{-(M+1)}}.
\end{multline}
With this notation \eqref{3-18} can be written in the form
\begin{multline}\label{3-21}
\dot{P}^{M,M}_{ij}(t)=-\left(a(M)+u(\alpha^j,1,M)\right)P_{ij}^{M,M}\\+\sum_{m=1}^{\infty}\left( u(\alpha^j,m,M)-u(\alpha^j,1,M)\right) P_{ij}^{M,M+m}.
\end{multline}
Let $i$ be fixed and set $\{\alpha^{j'}\}_{-(M+1)}=\{\alpha^j\}_{-(M+2)}\times \{\gamma\}$. Insert $j'$ for $j$ in \eqref{3-21} and sum the equations over $\gamma$. Since $0\leq \gamma\leq \Bmi{j}{M+2}-1$ we add $\Bmi{j}{M+2}$ equations.
We obtain
\begin{multline}\label{3-22}
\dot{P}^{M,M+1}_{ij}\\=-\left[ a(M)+u(\alpha^j,1,M)-\Bmi{j}{M+2}\left( u(\alpha^j,1,M)-u(\alpha^j,2,M)\right) \right]P^{M,M+2}_{ij}\\+\Bmi{j}{M+2}\sum_{m=2}^{\infty}\left( u(\alpha^j,m,M)-u(\alpha^j,m+1,M)\right)P^{M,M+m}_{ij}.
\end{multline}
In view of \eqref{3-7}, \eqref{3-10} and \eqref{3-16}
\begin{multline}\label{3-23}
a(M)+u(\alpha^j,1,M)-\Bmi{j}{M+2}(u(\alpha^j,1,M)-u(\alpha^j,2,M))\\=a(M+1)+u(\alpha^j,1,M+1).
\end{multline}
Thus \eqref{3-22} becomes
\begin{multline}\label{3-24}
\dot{P}^{M,M+1}_{ij}=-\left( a(M+1)+u(\alpha^j,1,M+1)\right) P_{ij}^{M,M+1}\\
+\Bmi{j}{M+2}\sum_{m=2}^{\infty}\left( u(\alpha^j,m,M)-u(\alpha^j,m+1,M)\right)P_{ij}^{M,M+m}.
\end{multline}
Iterating this procedure we obtain for $k\in \N$
\begin{multline}\label{3-25}
\dot{P}_{ij}^{M,M+k}=-\left(a(M+k)+u(\alpha^j,1,M+k)\right)P_{ij}^{M,M+k}\\
+\Bmi{j}{M+2}\dotsm \Bmi{j}{M+k+1}\sum_{m=k+1}^{\infty}\left(u(\alpha^j,m,M)-u(\alpha^j,m+1,M)\right)P_{ij}^{M,M+m}.
\end{multline}
After multiple application of \eqref{3-16} this can be written as
\begin{multline}\label{3-26}
\dot{P}_{ij}^{M,M+k}=-\left(a(M+k)+u(\alpha^j,1,M+k)\right)P_{ij}^{M,M+k}\\
+\sum_{m=1}^{\infty}\left(u(\alpha^j,m,M+k)-u(\alpha^j,m+1,M+k)\right)P_{ij}^{M,M+k+m}.
\end{multline}
We proved this formula for $k \in \N$, but it is also valid for $k=0$ in the sense that for $k=0$ it reduces to \eqref{3-21}. As a result of \eqref{3-4} and \eqref{3-20} we have
\begin{equation}\label{3-27}
P_{ij}^{M,M+k}(0)=\delta_{\{\alpha^i\}_{-(M+k+1)}\{\alpha^j\}_{-(M+k+1)}}.
\end{equation}
The coefficients $u$ on the right side of \eqref{3-26} depend on $\{\alpha^j\}_{-(M+k+2)}$. Thus the functions
$$P_{ij}^{M,M+k}=P_{\{\alpha^i\}_{-(M+1)}\{\alpha^{j'}\}_{-(M+k+1)}},$$
where
$$\{\alpha^{j'}\}_{-(M+k+2)}=\{\alpha^j\}_{-(M+k+2)}$$ i.e.\\ 
$$\{\alpha^{j'}\}_{-(M+k+1)}=\{\alpha^j\}_{-(M+k+2)}\times \{\gamma\}$$
satisfy the same system of equations for every value of $\gamma$:\\
 $0\leq \gamma \leq \Bmi{j}{M+k+2}-1$.
This together with the initial conditions \eqref{3-27} and the uniqueness of the solution yields
\begin{prop}\label{p3-2}
Let $P_{ij}^{M,M+k}$ be the solution of \eqref{3-26} satisfying the initial conditions \eqref{3-27} and $\{\alpha^{j'}\}_{-(M+k+1)}=\{\alpha^j\}_{-(M+k+2)}\times \{\gamma\}$ where $0\leq \gamma \leq \Bmi{j}{M+k+2}-1$.
If $\rho(\alpha^i,\alpha^j)=q^{M+k+1}$ then all functions $P_{ij}^{M,M+k}$ with $\gamma \neq \alpha^i_{-(M+k+1)}$ coincide. If $\rho(\alpha^i,\alpha^j)>q^{(M+k+1)}$ then all functions $P_{ij'}^{M,M+k}$ coincide and are equal to $P_{ij}^{M,M+k}$.
\qed\end{prop}
We also have
\begin{prop}\label{p3-3}
Put $N=M+k$.Then the functions $P_{ij}^{N,N}$ satisfy the equations 
\begin{multline}\label{3-28}
\dot{P}_{ij}^{N,N}=-\left(a(N)+u(\alpha^j,1,N)\right)P_{ij}^{N,N}\\
+\sum_{m=1}^{\infty}\left(u(\alpha^j,m,N)-u(\alpha^j,m+1,N)\right)P_{ij}^{N,N+M},
\end{multline}
identical with the equations \eqref{3-26} and the initial conditions 
\begin{equation}\label{3-29}
P_{ij}^{N,N}(0)=\delta_{ij}
\end{equation}
which coincide with the initial conditions \eqref{3-27}.  
\qed \end{prop}
It follows from \eqref{3-26} that
\begin{multline}\label{3-30}
\Bmi{j}{M+k+2}\dot{P}_{ij}^{M,M+k}-\dot{P}_{ij}^{M,M+k+1}\\
=-\left(a(M+k)+u(\alpha^j,1,M+k)\right)\left(\Bmi{j}{M+k+2}P_{ij}^{M,M+k}-P_{ij}^{M,M+k+1}\right).
\end{multline}
Let us concentrate on the case $i=j$. Then the solution of eq. \eqref{3-30} with the initial condition \eqref{3-27} reads
\begin{multline}\label{3-31}
\Bmi{i}{M+k+2}P_{ii}^{M,M+k}-P_{ii}^{M,M+k+1}\\
=\left(\Bmi{i}{M+k+2}-1\right)\exp\left\lbrace -\left(a(M+k)+u(\alpha^i,1,M+k)\right)t\right\rbrace ,
\end{multline}
or
\begin{multline}\label{3-32}
P_{ii}^{M,M+k}-\Bmi{i}{M+k+2}^{-1}P_{ii}^{M,M+k+1}\\
=\Bmi{i}{M+k+2}^{-1}\left(\Bmi{i}{M+k+2}-1\right)\exp\left\lbrace -\left(a(M+k)+u(\alpha^i,1,M+k)\right)t\right\rbrace.
\end{multline}
The following formula is a direct consequence of \eqref{3-32}
\begin{multline}\label{3-33}
\left(\Bmi{i}{M+k+n+1}\dotsm \Bmi{i}{M+k+2}\right)^{-1}P_{ii}^{M,M+k+n}\\
-\left(\Bmi{i}{M+k+n+2}\dotsm \Bmi{i}{M+k+2}\right)^{-1}P_{ii}^{M,M+k+n+1}\\
=\left(\Bmi{i}{M+k+n+2}\dotsm \Bmi{i}{M+k+2}\right)^{-1}\left(\Bmi{i}{M+k+n+2}-1\right)\\
\exp\left\lbrace -\left(a(M+k+n)+u(\alpha^i,1,M+k+n)\right)t\right\rbrace .
\end{multline}
This formula is valid for $n \in \N$, but we extend it for $n=0$ to be \eqref{3-32}. Summing equations \eqref{3-33} over $n$ from $n=0$ to $n=m-1$ we obtain
\begin{multline}\label{3-34}
P_{ii}^{M,M+k}=\left(\Bmi{j}{M+k+m+1}\dotsm \Bmi{i}{M+k+2}\right)^{-1}P_{ii}^{M,M+k+m}\\
+\sum_{n=0}^{m-1}\left(\Bmi{i}{M+k+n+2}\dotsm \Bmi{i}{M+k+2}\right)^{-1}\left(\Bmi{i}{M+k+n+2}-1\right)\\
\exp\left\lbrace - \left( a(M+k+n)+u(\alpha^i,1,M+k+n)\right)t\right\rbrace .
\end{multline}
The right hand side is split in two parts. The splitting depends on m. Since $\B{N}\geq 2$ the limit as $m \to \infty$ yields
\begin{multline}\label{3-35}
P_{ii}^{M,M+k}(t)=\sum_{n=0}^{\infty}\left(\Bmi{i}{M+k+n+2}\dotsm \Bmi{i}{M+k+2}\right)^{-1}\\
\left(\Bmi{i}{M+k+n+2}-1\right)\exp\left\lbrace - \left(a(M+k+n)+u(\alpha^i,1,M+k+n)\right)t\right\rbrace .
\end{multline}
Now we are ready to find the solution of \eqref{KCa} satisfying the initial conditions \eqref{3-4}. The function $P_{ii}^{M,M}$ is given by \eqref{3-35} with $k=0$ i.e.
\begin{multline}\label{3-36}
P_{ii}^{M,M}=\sum_{n=0}^{\infty}\left(\Bmi{i}{M+n+2}\dotsm \Bmi{i}{M+2}\right)^{-1}\left(\Bmi{i}{M+n+2}-1\right)\\
\exp\left\lbrace -\left(a(M+n)+u(\alpha^i,1,M+n)\right)t\right\rbrace.
\end{multline}
If $\rho(\alpha^i,\alpha^j)=q^{M+k}$, $k\in \N$ then by  multiple application of prop. \ref{p3-2} we obtain
\begin{multline}\label{3-37}
P_{ij}^{M,M}(t)=B^{-1}(\alpha^j,k,M)\left(P{ij}^{M,M+k}-P_{ij}^{M,M+k-1}\right)\\
=B^{-1}(\alpha^j,k,M)B^{-1}_{\{\alpha^i\}_{-(M+k+1)}}\left(\Bmi{i}{M+k+1}-1\right)\\
[ \sum_{n=0}^{\infty}\left( \Bmi{i}{M+k+n+2}\dotsm \Bmi{i}{M+k+2}\right)^{-1}\left( \Bmi{i}{M+k+n+2}-1\right)\\
\exp\left\lbrace -\left(a(M+k+n)+u(\alpha^i,1,M+k+n)\right)t\right\rbrace \\ - \exp\left\lbrace -\left(a(M+k-1)+u(\alpha^i,1,M+k-1)\right)t\right\rbrace].
\end{multline}
Formulas \eqref{3-36}, \eqref{3-37} complete our task to construct transition functions $P_{ij}^{M,M}$ for a class of processes on $\K(M)$. The next step of our discussion will be to construct the transition functions on $\SB$. This can be done as follows:\\
Let $\alpha,\beta \in \SB$. Then $\ami{M+1}=\{\alpha^i\}_{-(M+1)}$ and $\bmi{M+1}=\{\alpha^j\}_{-(M+1)}$ for some $i,j\in\N$. Set $P_{\ami{M+1}\bmi{M+1}}(t) = P_{ij}(t)$. Then by \eqref{3-36}, \eqref{3-37} we have
\begin{multline}
\label{3-38}
P_{\ami{M+1}\ami{M+1}}(t)\\
 = \sum_{n=0}^\infty \left(\Bmi{}{M+n+2} \dotsm \Bmi{}{M+2}\right)^{-1} \left( \Bmi{}{M+n+2}-1\right)\\
\times \exp \left\lbrace-\left(a(M+n)+u(\alpha,1,M+n)\right)t\right\rbrace
\end{multline}
and
\begin{multline}
\label{3-39}
P_{\ami{M+1}\bmi{M+1}}(t) = B^{-1}(\beta ,k,M) \Bmi{}{M+k+1}^{-1}\left( \Bmi{}{M+k+1}-1\right)\\
[\sum_{n=0}^\infty \left(\Bmi{}{M+k+n+2}\dotsm \Bmi{}{m+k+2}\right)^{-1}\left(\Bmi{}{M+k+n+2}-1\right)\\
\times \exp\left\lbrace -\left(a(M+k+n)+u(\alpha,1,m+k+n)\right)t\right\rbrace\\
-\exp\left\lbrace -\left(a(M+k-1)+u(\alpha,1,M+k-1)\right)t\right\rbrace],
\end{multline}
when $\rho (\alpha ,\beta ) = q^{M+k}$.
Let $M,N\in\Z$ and $M\leq N$. Then $\bmi{N+1}$ is an union of the balls of radius $q^M$ i.e.
\begin{equation}\label{3-40}
\bmi{N+1} = \bigcup_\gamma \bmi{N+1}\times \{\gamma_{-N}\}\times \dotsm \{\gamma_{-(M+1)}\},
\end{equation}
where the union runs over all B-products of $\bmi{N+1}$ and the $(N-M)$-tuples $\gamma$. Then we define
\begin{equation}\label{3-41}
P_{\ami{M+1}\bmi{N+1}}(t) = \sum_\gamma P_{\ami{M+1}\bmi{N+1}\times \{\gamma_{-N}\}\times \dotsm \times \{\gamma_{-(M+1)}\}}(t).
\end{equation}
By prop. \ref{p3-2} the function \eqref{3-41} does not depend on $M$ provided $M\leq N$. Since $\alpha = \cap_{M\leq N} \ami{M+1}$ we set
\begin{equation}
\label{3-42}
P(\alpha,\bmi{N+1},t) = P_{\ami{M+1}\bmi{N+1}}(t)
\end{equation}

We shall show in the next section that \eqref{3-42} defines the transition function for a stochastic process on $\SB$.
\section{Markovian Semigroup and its Generator}
In this section we shall define a Borel measure on $\SB$ and show that\\ $P(\alpha,\{\beta\}_{-(N+1)},t)$ discussed in the previous section defines a Markovian semigroup. We recall that $\K(M)$  defined by \eqref{2-18} is the family of all disjoint balls of radius $q^M$. Then
$$\K:=\bigcup_{M\in \Z}\K(M)$$ is the family of all balls in $\SB$. We define a set function $\mu$ on $\K$ as follows
\begin{equation}
\label{4-1}
\mu(\{0\}_{-1})=1,
\end{equation}
and
\begin{equation}\label{4-2}
\mu(\{\alpha\}_{-(M+1)})=\Bmi{}{M+1}\mu(\{\alpha\}_{-M})
\end{equation} 
for all $\alpha \in \SB$ and $M \in \Z$.\\ 
It follows from \eqref{4-2} that the numbers $\mu(\{\alpha\}_{{-(M+1)}\times \{\gamma\}})$, $0\leq \gamma \leq \Bmi{}{M+1}-1$ are equal. By standard arguments \cite{B} $\mu$ can be extended to a Borel measure on $\SB$. 
Similarly for any $\alpha \in \SB$ and $t>0$, $P(\alpha,\bmi{k+1},t)$ defines a set function on $\K$ and can be extended to a Borel measure on $\SB$. 
Given a ball $\{\beta\}_{-(k+1)}$ and $t>0$ then $P(\alpha,\{\beta\}_{-(k+1)},t)$ is a function of $\alpha \in \SB$ and by prop. \ref{p3-2} and \eqref{3-42} it is constant on every ball $\{\alpha\}_{-(k+1)}$. It follows that for any Borel set $A \subset{\SB}$, $P(\alpha,A,t)$ is a $\mu$-measurable function. Thus $P(\alpha,A,t)$ is a family of positive integral kernels.
For a real valued Borel function $u$ on ${\SB}$ put $p_{t}u(\eta)=\int_{\SB}P(\eta,d\xi,t)u(\xi)$ whenever the integral makes sense.
\begin{prop}\label{p4-1}
$P(\eta,A,t)$ has following properties 
\begin{enumerate}
\renewcommand{\theenumi}{(\roman{enumi})}
\item The integral kernel $P(\eta,A,t)$ is $\mu$-symmetric in the sense of Fukushima \cite{F},\cite{FOT} i.e. for any pair of nonnegative Borel functions $u, v$
\begin{multline}\label{4-3}
\int_{\SB}u(\eta)\left(\int_{\SB}v(\xi)P(\eta,d\xi,t)\right)\mu(d\eta)\\
=\int_{\SB}\left(\int_{\SB} u(\eta)P(\xi,d\eta,t)\right)v(\xi)\mu(d\xi)\leq \infty.
\end{multline}
\item $P(\eta,\SB,t)=1$ for all $\eta \in {\SB}$ and $t>0$.
\item $p_{t}(p_{s}u(\eta))=p_{t+s}u(\eta)$ for $t,s>0$ and $u$ any bounded  Borel function.
\end{enumerate} 
\end{prop}
\begin{proof} 
\begin{enumerate}
\renewcommand{\theenumi}{(\roman{enumi})}
\item It is sufficient to prove \eqref{4-3} for $u=\chi_{\{\alpha^i\}_{-(M+1)}}$, $v=\chi_{\{\alpha^j\}_{-(M+1)}}$. Then
$$\int_{\SB}\chi_{\{\alpha^j\}_{-(M+1)}}(\eta)P(\xi,d\eta,t)=P(\xi,\{\alpha^j\}_{-(M+1)},t).$$
If $\xi \in \{\alpha^i\}_{-(M+1)}$ then by \eqref{3-42}
$$P(\xi,\{\alpha^j\}_{-(M+1)},t)=P_{\{\alpha^i\}_{-(M+1)}\{\alpha^j\}_{-(M+1)}}(t)=P_{ij}(t).$$
Thus $$\int_{\SB}\chi_{\{\alpha^i\}_{-(M+1)}}(\xi)P(\xi,\{\alpha^j\}_{-(M+1)},t)\mu(d\xi)=\mu\left(\{\alpha^i\}_{-(M+1)}\right)P_{ij}(t).$$
Similarly
\begin{multline*}
\int_{\SB}\chi_{\{\alpha^j\}_{-(M+1)}}(\xi)\left(\int_{\SB}\chi_{\{\alpha^i\}}(\eta)P(\xi,d\eta,t)\right)\mu(d\xi)\\
=\mu(\{\alpha^j\}_{-(M+1)})P_{ji}(t).
\end{multline*}
If $\rho(\alpha^j,\alpha^i)=q^{M+k}$, $k\geq 1$, then according to \eqref{4-2} we have
$$\mu(\{\alpha^j\}_{-(M+k+1)})=\Bmi{j}{M+k+1}\dotsm \Bmi{j}{M+2}\mu(\{\alpha^j\}_{-(M+1)}),$$
and by \eqref{3-9} and the fact that $\{\alpha^j\}_{-(M+k+1)}=\{\alpha^i\}_{-(M+k+1)}$ we get
$$\mu(\{\alpha^j\}_{-(M+1)})=B^{-1}(\alpha^j,k,M)\dfrac{\Bmi{i}{M+k+1}-1}{\Bmi{i}{M+k+1}}\mu(\{\alpha^j\}_{-(M+k+1)}).$$
Similarly
$$\mu(\{\alpha^i\}_{-(M+1)})=B^{-1}(\alpha^i,k,M)\dfrac{\Bmi{i}{M+k+1}-1}{\Bmi{i}{M+k+1}}\mu(\{\alpha^j\}_{-(M+k+1)}).$$
Finally using \eqref{3-37} we obtain
\begin{multline*}
\mu\left(\{\alpha^i\}_{-(M+1)}\right)P_{ij}(t)=\\
B^{-1}(\alpha^i,k,M)B^{-1}(\alpha^j,k,M)\left(\dfrac{\Bmi{i}{M+k+1}-1}{\Bmi{i}{M+k+1}}\right)^2A_i(t)\mu(\{\alpha^j\}_{-(M+k+1)}),
\end{multline*}
\begin{multline*}
\mu\left(\{\alpha^j\}_{-(M+1)}\right)P_{ji}(t) \\ =B^{-1}(\alpha^j,k,M)B^{-1}(\alpha^i,k,M)\left(\dfrac{\Bmi{i}{M+k+1}-1}{\Bmi{i}{M+k+1}}\right)^2A_j(t)\mu(\alpha^j\}_{-(M+k+1)}),
\end{multline*}
where $A_i(t)$ resp. $A_j(t)$ is the time dependent factor in the square bracket in \eqref{3-37}. We conclude the proof observing that $A_i(t)=A_j(t)$ and thus
$$\mu(\{\alpha^i\}_{-(M+1)})P_{ij}(t)=\mu(\{\alpha^j\}_{-(M+1)})P_{ji}(t).$$

\item We have \mbox{$P(\eta,\SB,t)=\lim \limits_{k\to \infty}P_{ii}^{M,M+k}(t)$} where $\eta \in K_{i}^{M}$. 
Taking into account that
$$u(\alpha^j,1,N)=(\Bmi{j}{N+2}-1)^{-1}(a(N)-a(N+1))\leq a(N)-a(N+1)$$
and using \eqref{3-35} we find
\begin{multline*}
P_{ii}^{M,M+k}(0)=\sum_{n=0}^{\infty}(\Bmi{i}{M+k+n+2}\dotsm \Bmi{i}{M+k+2})^{-1}(\Bmi{i}{M+k+n+2}-1)\\
=1 \geq P_{ii}^{M,M+k}(t)\geq \exp\{-(2a(M+k)-a(M+k+1))t\}\underset{k\to \infty}{\to}1.
\end{multline*}
\item It is sufficient to prove the statement for $u=\chi_{\{{\alpha}^i\}_{-(M+1)}}.$ The functions ${P_{ij}}^{M,M}(t)$
defined by \eqref{3-36},\eqref{3-37} solve the Chapman-Kolmogorov equations \eqref{KC} and hence posses the semigroup property
\begin{equation}\label{4-4}
\sum_{l=1}^{\infty}P_{il}^{M,M}(s)P_{lj}^{M,M}(t)=P_{ij}^{M,M}(s+t).
\end{equation}
In the notation of Section 3 \eqref{4-4} reads
\begin{equation}\label{4-5}
\int_{\SB}P_{s}(\xi,d\eta)P_{t}(\eta,K_{j}^{M})={\sum}_{l=1}^{\infty} \int_{K_{l}^{M}}P_{s}(\xi,d\eta)P(\eta,K_{j}^{M})=P_{s+t}(\xi,K_{j}^{M}).
\end{equation} 
Let $M\leq N\in{\mathbf{Z}}$. Then by Propositions \ref{p3-2} and \ref{p3-3} $P_{t}(\xi,K_{j}^N)$ can be defined equivalently either by \eqref{3-42} with the right hand side given by \eqref{3-41} or by
\begin{equation}\label{4-6}
P_{t}(\xi,K_{j}^{N})=P_{ij}^{N}(t), \xi\in K_{i}^{N},
\end{equation}
where $P_{ij}^{N}(t)$ is a solution of the Chapman-Kolmogorov equation over the state space $\mathcal{K}^{N}.$
Consequently \eqref{4-5} holds for any $M\in{\Z}.$
\end{enumerate}
\end{proof}
As an immediate consequence of proposition \ref{p4-1} $p_t$ extends uniquely to a self adjoint Markovian semigroup $T_t,t>0$ acting in $L^2(\SB,\mu)$.
\begin{prop}\label{p4-2}
The Markovian semigroup $T_t,t>0$ acting in $L^2(\SB,\mu)$ defined by $p_t,t>0$ is strongly continuous.\\
\end{prop}
\begin{proof} Due to the contractivity of $T_t$ it will be sufficient to show that \mbox{$\lim \limits_{t \downarrow 0}\Vert T_tf-f\Vert=0$}
 for $f$ of the form $f=\sum_{i=1}^{n}f_i\chi_{K_i}$. But using the self-adjointness of $T_t$ and the initial condition $P_{ij}(0)=\delta_{ij}$ we obtain indeed
\begin{multline*}
\Vert T_tf-f\Vert^2-(\sum_i f_i \chi_{K_i},T_{2t} \sum_j f_j \chi_{K_j})-2(T_{t}\sum_i f_i \chi_{K_i},\sum_j f_j \chi_{K_j})\\
 +(\sum_i f_i \chi_{K_i},\sum_j f_j \chi_{K_j})=\sum_{i,j} f_i f_j (\mu(K_i) P_{ij}(2t)-2 \mu(K_j) P_{ji}(t)\\
 +\delta_{ij} \mu(\chi_{K_j}))\underset{t\downarrow 0}{\to} \sum_i f^2_i( \mu(K_i)-2\mu(K_i)+\mu(K_i))=0.
\end{multline*}\end{proof}

To summarize $T_t,t>0$ is a strongly continuous self adjoint contraction semigroup acting in $L^{2}({\SB},\mu)$ and hence it has the representation
$$T_t=e^{-Ht},t\geq0$$ where $H$ is a non-negative self-adjoint operator acting in $L^{2}({\SB},\mu)$. Let $f\in L^2(\SB,\mu).$ Then by definition
$$(Hf)(\eta)=\lim_{t\downarrow 0} t^{-1}\left[f(\eta)-\left(T_tf\right)(\eta)\right]=\lim_{t \downarrow 0}t^{-1}\left[f(\eta)-\int_{\SB}f(\xi)P(\eta,d\xi,t)\right]$$
whenever the strong limit exists. To find the explicit formula for $H$ we proceed as follows. Take $f=\chi_{\{\alpha\}_{-(M+1)}}$. If $\eta\in \{\alpha\}_{-(M+1)}$ then
$$H\chimi{M+1}(\eta)=\lim_{t \downarrow 0}t^{-1}\left[1-P(\eta,\{\alpha\}_{-(M+1)},t)\right]=a(M).$$
Indeed
\begin{multline*}
\frac{1}{t}\left[1-P(\eta,\{\alpha\}_{-(M+1)})\right]=\\
\sum_{n=0}^{\infty}\left(\Bmi{}{M+n+2} \dotsm \Bmi{}{M+2}\right)^{-1}\left(\Bmi{}{M+n+2}-1\right)\\
\times\left(1-\exp\left\{-\left(a(M+n)+u(\alpha^i,1,M+n)\right)t\right\}\right)\frac{1}{t}\\
\underset{t\to 0}{\to} \sum_{n=0}^{\infty}\left(\Bmi{}{M+n+2} \dotsm \Bmi{}{M+2}\right)^{-1}\left(\Bmi{}{M+n+2}-1\right)\\
\times \left[ a(M+n)+\left(\Bmi{}{M+n+2}-1\right)^{-1}\left(a(M+n)-a(M+n+1)\right)\right]\\
=\sum_{n=0}^{\infty}\left(\Bmi{}{M+n+2} \dotsm \Bmi{}{M+2}\right)^{-1}\left(\Bmi{}{M+n+2}a(M+n)-a(M+n+1)\right)\\=a(M).
\end{multline*}
If $\rho(\eta,\alpha)=q^{M+k},k\in\N$ then we have by \eqref{3-39}
\begin{multline*}
H\chimi{M+1}(\eta)=-\lim_{t\downarrow 0}t^{-1}P(\eta,\{\alpha\}_{-(M+1)},t)\\
=B^{-1}(\alpha,k,M)\Bet{M+k+1}^{-1}\left(\Bet{M+k+1}-1\right)\\
\times [\sum_{n=0}^{\infty}\left(\Bet{M+k+n+2} \dotsm \Bet{M+k+2}\right)^{-1}\left(\Bet{M+k+n+2}-1\right) \\
\times\left(a(M+k+n)+u(\eta,1,M+k+n)\right)-\left(a(M+k-1)+u(\eta,1,M+k-1)\right)]\\
=-B^{-1}(\alpha,k,M)\left(a(M+k-1)-a(M+k)\right).
\end{multline*}
Thus we have proved the formula
\begin{equation}\label{4-7}
H\chimi{M+1}(\eta)=\left\{\begin{array}{l}
a(M)\ if\ \eta \in \{\alpha\}_{-(M+1)},\\
-B^{-1}(\alpha,k,M)\left(a(M+k-1)-a(M+k)\right)\\
if\ \rho(\eta,\alpha)=q^{M+k}.
\end{array}\right.
\end{equation} 
This formula is valid for all $\alpha \in \SB$ and $M\in\Z$. Note that by \eqref{4-2} 
$$
\mu (\{\alpha\}_{-(M+k+1)}\setminus \{\alpha\}_{-(M+k)})=B(\alpha,k,M)\mu (\{\alpha\}_{-(M+1)}).
$$ 
Thus 
\begin{multline}
\left\|H\chimi{M+1}\right\|^{2}\\
=\mu(\{\alpha\}_{-(M+1)})[a(M)^{2}+\sum_{k=1}^{\infty}B^{-1}(\alpha,k,M)[a(M+k-1)-a(M+k)]^{2}]
\end{multline}
is finite because $B^{-1}(\alpha,k,M)\longrightarrow 0$ and $a(N)\longrightarrow 0$ as $N\longrightarrow \infty.$
Consequently the characteristic functions of the balls in $\SB$ belong to the domain $D(H)$ of $H$. 
The spectral properties of $H$ are described by
\begin{thm}\label{p4-3}
Let $-H$ denote as above the generator of the strongly continuous semigroup $T_t$ with the kernel defined by \eqref{3-42}. \\Then
\begin{enumerate}
\renewcommand{\theenumi}{(\alph{enumi})}
\item\label{p4-4a} For any $M\in\Z$ such that $a(M)>0$ and $\alpha\in\SB$ there corresponds an eigenvalue $h_{M,\alpha}$ of $H$ given by
\begin{equation}\label{4-8}
h_{M,\alpha}=\left(\Bmi{}{M+2}-1\right)^{-1}\left(\Bmi{}{M+2}a(M)-a(M+1)\right)
\end{equation}
and a $\Bmi{}{M+2}-1$ dimensional eigenspace spanned by vectors of the form
\begin{equation}\label{4-9}
e_{M,\alpha}=\sum_{\gamma=0}^s b_{\gamma} {\chi}_{\{\alpha\}_{-(M+2)}\times\{\gamma\}},
\end{equation}
where
\begin{equation}\label{4-10}
\sum_{\gamma=0}^sb_{\gamma}=0\ and\ s=\Bmi{}{M+2}-1.
\end{equation}
If $a(M)=0$ then $\chimi{M+1}$ is an eigenvector of $H$ to the eigenvalue 0.
\item\label{p4-4b} The linear hull spanned by the vectors $e_{M,\alpha}$, $M\in\Z$, $\alpha\in\SB$ is dense in $L^2(\SB,\mu)$.
\end{enumerate}
\end{thm}
\begin{proof}The proof of  \ref{p4-4a} follows the proof of an analogous statement in \cite{AKa} with minor changes:\\ 

 Let $a(M)>0$ and put \mbox{$e_{M,\alpha}=\sum_{\gamma=0}^s b_{\gamma}{\chi}_{\{\alpha\}_{-(M+2)}\times\{\gamma\}}$}
 where $b_{\gamma}\in\R$ and $s=\Bmi{}{M+2}-1$.
Then $\supp e_{M,\alpha} \subset \{\alpha\}_{-(M+2)}$.\\
Let $\rho(\xi,\alpha)=q^{M+k+1},\ k\in\N$. Then by \eqref{4-7}

\begin{equation}\label{4-11}
He_{M,\alpha}(\xi)=-B^{-1}(\alpha,k+1,M)(a(M+k)-a(M+k+1))\sum_{\gamma=0}^{s}b_{\gamma}.
\end{equation}
This is zero if either $a(M+1)=0$ and then all $a(M+k)=0$ or\\
 $\sum^s_{\gamma=0}b_{\gamma}=0$. In either case we have $\supp\left(He_{M,\alpha}\right)\subset\{\alpha\}_{-(M+2)}$. The condition
\begin{equation}\label{4-12}
He_{M,\alpha}=h_{M,\alpha}e_{M,\alpha}
\end{equation}
is equivalent to the system of algebraic equations
\begin{equation}\label{4-13}
\sum_{\gamma=0}^s a_{\beta \gamma}b_{\gamma}=0,\ \beta=0,\dotsc ,s
\end{equation}
where $a_{\beta\beta}=a(M)-h_{M,\alpha}\ ,\ \beta=0,\dotsc ,s$\\ and $a_{\beta\gamma}=-B^{-1}(\alpha,1,M)\left(a(M)-a(M+1)\right)$.
If $a(M+1)=0$ then \eqref{4-13} is solved by
$b_0=b_1=\dotsc =b_s \text{ with } h_{M,\alpha}=0$
or by $b_{\gamma}$ satisfying
\begin{equation}\label{4-14}
\sum_{\gamma}^s b_{\gamma}=0
\end{equation}
with
\begin{equation}\label{4-15}
h_{M,\alpha}=\left(\Bmi{}{M+2}-1\right)^{-1}\left(\Bmi{}{M+2}a(M)-a(M+1)\right).
\end{equation}
If $a(M+1)>0$ then \eqref{4-11} vanish only if \eqref{4-14} holds. In this case \eqref{4-13} is solved with $h_{M,\alpha}$ given by \eqref{4-15}. 
If $a(M)=0$ then $a(M+k)=0$ for all $k\in N$ and by \eqref{4-7} $\chimi{M+1}$ is an eigenvector of $H$ to the eigenvalue 0. This proves part \ref{p4-4a}.\\
To prove (b) it is sufficient to show that for any $\beta\in\SB$ and $\N\in\Z$ the function $\chi_{\{\beta\}_{-(N+1)}}$ can be approximated by the vectors $e_{M,\alpha}$. According to \ref{p4-4a} any pair $\alpha\in\SB,\ M\in\Z$ defines an eigenvalue $h_{M,\alpha}$ and a corresponding $\Bmi{}{M+2}-1$ dimensional eigenspace.\\
For $k=0,1,\dotsc ,\Bmi{}{M+2}-2$ set
$$b^k_{\gamma,M}=\left\{\begin{array}{l}
1;\ \gamma=\alpha_{-(M+1)},\alpha_{-(M+1)}+1,\dotsc ,\alpha_{-(M+1)}+k\\
-(k+1);\ \gamma=k+1+\alpha_{-(M+1)}\\
0;\text{ otherwise}
\end{array}\right.$$
$\gamma$ is taken modulo $\Bmi{}{M+2}$. Then $\sum_{\gamma=0}^s b_{\gamma,M}^k=0$ and the vectors
$$e_{M,\alpha}^k =\sum_{\gamma=0}^s b_{\gamma,M}^k{\chi}_{\{\alpha\}_{-(M+2)}\times\{\gamma\}}$$
are pairwise orthogonal. We then have by \eqref{4-2}
\begin{gather*}
\left\Vert e_{M,\alpha}^k\right\Vert^2=\mu\left(\{\alpha\}_{-(M+2)}\right)\Bmi{}{m+2}^{-1}(k+1)(k+2)\\
\text{and } \left(\chimi{N+1}, e_{M,\alpha}^k\right)=\mu\left(\{\alpha\}_{-(N+1)}\right)\text{ for } M\geq N.
\end{gather*}
\begin{multline*}
\text{Thus }\sum_{M=N}^{\infty}\sum_{k=0}^{s_M}\left\vert\left(\chimi{N+1},e_{M,\alpha}^k\right)\left\Vert e_{M,\alpha}^k\right\Vert^{-1}\right\vert^2 \\
\shoveleft =\left(\mu\left(\{\alpha\}_{-(N+1)}\right)\right)^2\sum_{M=N}^{\infty}\left(\mu\left(\ami{M+2}\right)\right)^{-1}\Bmi{}{M+2}\sum_{k=0}^{s_M}\tfrac{1}{(k+1)(k+2)}\\
\shoveleft =\left(\mu\left(\ami{N+1}\right)\right)^2 \sum_{M=N}^{\infty}\left(\mu\left(\ami{M+2}\right)\right)^{-1}\left(\Bmi{}{M+2}-1\right)\\
\shoveleft =\mu\left(\ami{N+1}\right)\sum_{M=N}^{\infty}\left(\Bmi{}{M+2}\dotsm \Bmi{}{N+2}\right)^{-1}\left(\Bmi{}{M+2}-1\right)\\
\shoveleft =\mu\left(\ami{N+1}\right)=\left\Vert \chimi{N+1}\right\Vert^2.\\
\end{multline*}

\end{proof}
We have seen that $\chimi{M+1}\in D(H)$ for all $\alpha \in {\SB}$ and $M\in {\mathbf{Z}}.$ Put $D_0$ for the linear hull spanned by all $\chimi{M+1}.$ Then $D_{0}\subset D(H).$
Indeed, it follows from Theorem 4.3, b) that the eigenfunctions $e_{M,\alpha}^{k}\left\|e_{M,\alpha}^{k}\right\|^{-1}$ of $H$ form an orthonormal basis in $L^{2}(\SB,\mu).$ Since $e_{M,\alpha}^{k}$ are defined as linear combinations of the characteristic functions for the balls we have $e_{M,\alpha}^{k}\in D_{0}$ and consequently $D_{0}$ is dense in $D(H)$ in the graph norm. Thus we have proved
\begin{cor}\label{p4-4}
$D_0$ is an operator core for $H$ in $L^{2}({\SB},\mu).$
\end{cor}  
Let as before $-H$ be the generator of the strongly continuous Markov semigroup $T_t,\ t>0$ constructed above. Then
$$\E(f,g)=\left(H^\frac{1}{2}f,H^\frac{1}{2}g\right)$$
defined for all $f,g\in D[\E]=D(H^\frac{1}{2})$ is according to \cite{F}, \cite{FOT} a closed, symmetric Markovian quadratic form i.e. a Dirichlet form in $L^{2}({\SB},\mu)$.\\
Since $D_0$ is a core for $H$ it is also a core for $H^\frac{1}{2}$ i.e.
\begin{enumerate}
\renewcommand{\theenumi}{(\alph{enumi})}
\item\label{alpha} $D_0$ is dense in $D[\E]$ in the norm $\left(\E_1 (\cdot , \cdot )\right)^\frac{1}{2}=\left[\E (\cdot , \cdot )+(\cdot , \cdot )\right]^\frac{1}{2}$.\\
\end{enumerate}
Put $C_0(\SB)$ for the space of real valued continuous functions of compact support on $\SB$. Then by the Weierstrass-Stone theorem
\begin{enumerate}
\renewcommand{\theenumi}{(\alph{enumi})}
\refstepcounter{enumi}
\item\label{beta}$D_0$ is dense in $C_0(\SB)$ in the uniform norm topology.
\end{enumerate}
In the Fukushima terminology a set $D_0\subset D[\E]\cap C_0(\SB)$ enjoying properties \ref{alpha} and \ref{beta} is called a core for $\E$ and a Dirichlet form which has such a core is called regular. The regular Dirichlet forms can be expressed uniquely in terms of the Beurling-Deny representation. Thus we have the representation
\begin{multline}\label{4-16}
\E(u,v)=\E^{(c)}(u,v)+\underset{\SB\times\SB\setminus d}{\int \int}(u(\eta)-u(\xi))(v(\eta)-v(\xi))J(d\eta,d\xi)\\
+\int_{\SB} u(\xi)v(\xi)k(d\xi)
\end{multline}
for $u,v\in D[\E]\cap C_0(\SB)$.\\
Here $\E^{(c)}$ is a symmetric form satisfying $\E^{(c)}(u,v)=0$ if $v$ is constant on a neighborhood of $\supp u$, $J(d\eta,d\xi)$ is a symmetric positive Borel measure on $\SB\times\SB$ off the diagonal $d$ and $k$ a positive Borel measure on $\SB$. It turns out however that in our case $\E^{(c)}$ and $k$ vanish identically. Indeed put $u=\chimi{M+1}$ and $v=\chimibet{M+1}.$ Then for all $\alpha  ,\beta \in \SB$ and $M\in\Z$ the function $\chimibet{M+1}$ is constant on $\ami{M+1}=\supp\chimi{M+1}$. By Proposition \ref{p2-6} $\ami{M+1}$ is open and thus a neighborhood of itself. Accordingly
$$\E^{(c)}\left(\chimi{M+1},\chimibet{M+1}\right)=0$$
and consequently $\E^{(c)}$ vanishes identically.\\
Let $\rho(\alpha,\beta)=q^{M+k},\ M\in\Z,\ k\in\N$. Then $$\supp \chimi{M+1}\cup\supp \chimibet{M+1}=\emptyset$$ and
\begin{multline*}
\E\left(\chimi{M+1},\chimibet{M+1}\right)\\
\shoveleft = \underset{\SB\times\SB\setminus d}{\int\int}\left(\chimi{M+1}(\xi)-\chimi{M+1}(\eta)\right)\left(\chimibet{M+1}(\xi)-\chimibet{M+1}(\eta)\right)J(d\xi,d\eta)\\
\shoveleft = -2\underset{\SB\times\SB\setminus d}{\int\int}\chimi{M+1}(\xi)\chimibet{M+1}(\eta)J(d\xi,d\eta)\\
\shoveleft = -2J\left(\ami{M+1},\bmi{M+1}\right).\\
\end{multline*}
On the other hand by \eqref{4-7}
\begin{multline*}
\E\left(\chimi{M+1},\chimibet{M+1}\right)=\left(H^\frac{1}{2}\chimi{M+1},H^\frac{1}{2}\chimibet{M+1}\right)\\
=\left(\chimi{M+1},H{\chimibet{M+1}}\right)=-\mu\left(\ami{M+1}B^{-1}(\beta,k,M)\left(a(M+k+1)-a(M+k)\right)\right)
\end{multline*}
Thus \begin{multline}\label{4-17}
J\left(\ami{M+1},\bmi{M+1}\right)=\tfrac{1}{2}\mu\left(\ami{M+1}\right)\\
\times B^{-1}(\beta,k,M)\left(a(M+k+1)-a(M+k)\right)
\end{multline} which determines the measure $J$ uniquely.
\begin{rem*}
The formula \eqref{4-17} seems to contradict the symmetry of $J$. However a direct computation shows that
$$\mu\left(\ami{M+1}\right)B^{-1}(\beta,k,M)=\mu\left(\bmi{M+1}\right)B^{-1}(\alpha,k,M)$$
so that after all $J$ is symmetric.
\qed \end{rem*}
To see that $k=0$ observe that by \eqref{4-7}
$$\E\left(\chimi{M+1},\chimi{M+1}\right)=a(M)\mu\left(\ami{M+1}\right).$$
Since $\E^{(c)}=0$ it follows from \eqref{4-16} that
\begin{multline}\label{4-18}
\int_{\SB}\chimi{M+1}(\xi)k(d\xi)=k\left(\ami{M+1}\right)=a(M)\mu\left(\ami{M+1}\right)\\
-\underset{\SB\times\SB}{\int\int}\left\vert\chimi{M+1}(\xi)-\chimi{M+1}(\eta)\right\vert^2J(d\xi,d\eta)
\end{multline}
To compute the latter integral we observe that the integrand is different from zero (and then equals 1) only if either $\xi\in\ami{M+1}$ and $\eta\not\in\ami{M+1}$ or $\eta\in\ami{M+1}$ and $\xi\not\in\ami{M+1}$.Thus
\begin{multline}\label{4-19}
\underset{\SB\times\SB\setminus d}{\int\int}\left\vert\chimi{M+1}(\xi)-\chimi{M+1}(\eta)\right\vert^2 J(d\xi,d\eta)\\
=2\sum_{k=1}^\infty\sum_\beta J\left(\ami{M+1},\bmi{M+1}\right),
\end{multline}
where the second sum runs over all the balls $\bmi{M+1}$ such that $\rho(\alpha,\beta)=q^{M+k}$. We have then
\begin{multline*}
\sum_\beta J\left(\ami{M+1},\bmi{M+1}\right)=\sum_\beta\tfrac{1}{2}\mu\left(\ami{M+1}\right)\\
\times B^{-1}(\beta,k,M)\left(a(M+k+1)-a(M+k)\right)=\tfrac{1}{2}\mu\left(\ami{M+1}\right)\left(a(M+k+1)-a(M+k)\right).
\end{multline*}
and \eqref{4-19} equals
$$\mu\left(\ami{M+1}\right)\sum_{k=1}^\infty\left(a(M+k+1)-a(m+k)\right)=a(M)\mu\left(\ami{M+1}\right),$$
and by \eqref{4-18} the measure $k$ must vanish. As is well known the three terms in the Beurling-Deny representation have an interpretation in terms of corresponding stochastic process. Namely, $\E^{(c)}$ is associated with a diffusion process, $J$ the jump measure for a jump process and $k$ the killing measure. Our discussion shows that the process defined by the transition function \eqref{3-42}  constructed in section 3 is a purely jump process. The fact that the killing measure vanish is a result of the construction we applied. For similar construction on $p$-adics yielding also the killing measure see \cite{Myst}. The absence of the diffusion part is a consequence of the fact that $\SB$ is totally disconnected.\\
\textbf{Acknowledgments}.
The second author gratefully acknowledges financial support from SFB611 (Bonn) as well as the hospitality of the Institute of Applied Mathematics, University of Bonn.


\begin{thebibliography}{99}
\bibitem{A} Albeverio, S., Theory of Dirichlet Forms and Applications, pp. 4-106 in: S. Albeverio, W. Schachermayer, M. Talagrand, Lecture on Probability Theory and Statistics, St Flour 2000, Ed. P. Bernard. Lecture Notes in Maths. 1816, Springer, Berlin (2003).  
\bibitem{AKa} Albeverio, S., Karwowski, W., A random walk on p-adics: the generator and its spectrum. Stochastic processes and their Applications \textbf{53} (1994)1-22.
\bibitem{AKaZ} Albeverio, S., Karwowski, W., Zhao, X.: Asymptotic and spectral results for random walks on p-adics. Stochastic Processes and their Applications \textbf{83} (1999) 39-59.
\bibitem{AKhSh1} Albeverio, S., Khrennikov, Yu., Shelkovich, V.M.: Associated homogeneous p-adic distributions, J.Math.Anal. Appl. \textbf{313} (2006) 64-83. 
\bibitem{AKhSh2} Albeverio, S., Khrennikov, Yu., Shelkovich, V.M.: p-adic Colombeau-Egorov type theory of generalized functions. Math. Nachr. \textbf{278}(2005) 3-16.
\bibitem{AKhSh3} Albeverio, S., Khrennikov, Yu., Shelkovich, V.M.: Nonlinear problems in p-adic analysis: Associative algebras of p-adic distributions. Izviestia Akademii Nauk, Seria Math., \textbf{69} (2005) 221-263.
\bibitem{AKhS4} Albeverio, S. Khrennikov, A. Yu., Shelkovich, V. M.: Theory of $p$-adic distributions, linear and nonlinear, book subm. for publication.
\bibitem{ASTa} S. Albeverio, W. Schachermayer, M. Talagrand, Lecture on Probability Theory and Statistics, St Flour 2000, Ed. P. Bernard. Lecture Notes in Maths. 1816, Springer, Berlin (2003).  
\bibitem{ATi} Albeverio, S., Tirozzi, B.: An introduction to the mathematical theory of neural networks, in: P. Garrido and I. Marro (Eds.) Proceedings of "Fourth Granada Lectures in Computational Physics" vol. 493, of Lecture Notes in Physics. Springer Verlag, Berlin-New York-Heidelberg, 1997, pp. 197-222.
\bibitem{AZX} Albeverio, S., Zhao, X.: A remark on the relations between different constructions of random walks on $p$-adics. Markov Processes and Related Fields \textbf{6}, (2000)239-256.
\bibitem{AlEv} Aldous, D., Evans, S.: Dirichlet forms on totally disconnected spaces and bipartite Markov chains. J. Theor. Prob. \textbf{12} (1999) 839-857.
\bibitem{B} Bourbaki, N., El\'{e}ments de Math\'{e}matique Livre VI Int\'{e}gration, Hermann, Paris VI, 1969.
\bibitem{DelTa} Del Muto, M., Fig\`{a} Talamanca, A.: Diffusion on locally compact ultrametric spaces, Exp. Math. \textbf{22} (2004) 197-211.
\bibitem{E1} Evans, S.N.: Local properties of Levy processes on totally disconnected groups. J. Theoret. Prob. \textbf{2} (1989) 209-259.
\bibitem{E2} Evans, S.N.: Local fields Gausian measure. In Seminar on Stochastic Processes 1988. Birkhauser, Basel, pp. 121-60.
\bibitem{F} Fukushima, M., Dirichlet Forms and Markov Processes, North Holland, Amsterdam and Kodansha, Tokyo 1980.
\bibitem{FOT} Fukushima, M., Oshima, Y., Takeda, M., Dirichlet Forms and Symmetric Markov Processes, De Gruyter, Berlin 1994.
\bibitem{Kan} Kaneko, H.: A class of spatially inhomogeneous Dirihlet spaces on the p-adic number field. Stochastic Processes and their Applications \textbf{88}(2000) 161-174.
\bibitem{KaViM} Karwowski, W., Vilela Mendes, R.: Hierarchical structures and asymmetric processes on p-adics and adeles. J. Math. Phys. \textbf{35} (1994) 4637-4650.
\bibitem{Kh1} Khrennikov, A.: $p$-adic valued distributions in mathematical physics. Kluver Academic Publ., Dordreht, 1994.
\bibitem{Kh2} Khrennikov, A.: $p$-adic discrete dynamical systems and their applications in physics and cognitive science, Russian Journal of Mathematical Physics, \textbf{11} (2004) 45-70.  
\bibitem{Kh3} Khrennikov, A.: Information dynamics in cognitive, psychological, social and anomalous phenomena, Kluver Academic Publ., Dortrecht, 2004. 
\bibitem{Koc} Kochubei, A. N.: Pseudo-differential equations and stochastics over non-archimedean fields. Marcel Dekker, Basel (2001).  
\bibitem{Ko} Koblitz, N., $p$-adic numbers, $p$-adic analysis and zeta functions, 2nd ed. Springer, New York 1984.  
\bibitem{LiViM} Lima, R., Vilela Mendes, R.: Stochastic processes for the turbulent cascade. Physics Revue E \textbf{53} (1996) 3536-3540.
\bibitem{MaR} Ma, Z., M., R{\"o}ckner, M.: Introduction to the theory of (non-symmetric) Dirichlet forms. Springer-Verlag, Berlin-Heidelberg-New York, 1992.
\bibitem{Myst} Mystkowski, M.: Random walk on $p$-adics with non-zero killing part. Rep. Math. Phys. \textbf{34}, (1994) 133-141.
\bibitem{Pa} Pagliacci, M., Applications of diffusion processes on trees to mathematical finance, pp. 157-159 in Harmonic functions on trees and buildings, Contemporary Mathematics 206, AMS, Providence 1997.
\bibitem{RTV} Rammal, R., Toulouse, G., Virasoro, M. A., Ultrametricity for Physicists. Rev. Mod. Phys. \textbf{58} 
(1986) 765 - 788. 
\bibitem{Sturm} Sturm, K., Th.: Diffusion processes and heat kernels on metric spaces. Ann. Prob. \textbf{26} (1998) 1-55.
\bibitem{VlaVoZ} Vladimirov, V., Volovich, I., Zelnov, E.: p-adic numbers in mathematical physics. World Scientific, Singapore 1993.
\bibitem{Y} Yasuda, K.: Semi stable processes on local fields. Tohoku Math. Journal \textbf{58} 419-431.  
\end{thebibliography}
\end{document}